\documentclass[11pt]{article}
\usepackage[margin=1in]{geometry}
\usepackage{amsmath,amssymb,amsthm}
\usepackage{enumitem}
\usepackage{microtype}
\usepackage{hyperref}

\usepackage{makecell}
\usepackage{tablefootnote}

\usepackage[ruled,vlined,linesnumbered]{algorithm2e}
\SetKwComment{KwComment}{$\triangleright$\ }{}
\usepackage{tikz}
\usepackage{authblk}
\usepackage[numbers]{natbib}
\usetikzlibrary{arrows.meta,positioning,calc,fit,backgrounds,decorations.pathreplacing}

\newtheorem{definition}{Definition}
\newtheorem{theorem}{Theorem}
\newtheorem{corollary}{Corollary}[theorem]
\newtheorem{lemma}{Lemma}
\newcommand{\remove}[1]{}

\usepackage[ruled,vlined]{algorithm2e}

\title{Computing Least Fixed Points with Overwrite Semantics in Parallel and Distributed Systems}

\author[a]{Vijay K. Garg}
\author[b]{Rohan Garg}
\affil[a]{Department of Electrical and Computer Engineering, University of Texas at Austin}
\affil[b]{Department of Computer Science, Purdue University.

Emails: \texttt{{garg@ece.utexas.edu, rohang@purdue.edu}}}
\date{} 

\begin{document}
\bibliographystyle{alpha}
\maketitle

\begin{abstract}

We present methods to compute least fixed points of multiple monotone 
inflationary functions in parallel and distributed settings. While the classic
Knaster-Tarski theorem addresses a single function with sequential iteration, 
modern computing systems require parallel execution with overwrite semantics, 
non-atomic updates, and stale reads.
We prove three convergence theorems under progressively relaxed synchronization:
(1) Interleaving semantics with fair scheduling,
(2) Parallel execution with update-only-on-change semantics (processes write only on those
    coordinates whose values change), and 
(3) Distributed execution with bounded staleness (updates propagate within 
     $T$ rounds) and $i$-locality (each process modifies only its own component).

Our approach differs from prior work in fundamental ways: Cousot-Cousot's
chaotic iteration uses join-based merges that preserve information~\cite{cousot1977sigplan}. 
Instead, we use coordinate-wise overwriting.
Bertsekas's asynchronous methods assume contractions~\cite{bertsekas1983}. We use coordinate-wise 
overwriting with structural constraints (locality, bounded staleness) instead.
Applications include parallel and distributed algorithms for 
the transitive closure,
stable marriage, shortest 
paths, and fair division with subsidy problems. Our results provide the first exact least-fixed-point 
convergence guarantees for overwrite-based parallel updates without join 
operations or contraction assumptions.

\end{abstract}
\noindent\textbf{Keywords:}
fixed points, distributive lattices, parallel algorithms, asynchronous computation

\section{Introduction}

Fixed-point computation \citep{tarski1955,kleene1952,knaster1928} lies at the heart of many areas in computer science: logic, automata and optimization.
Given a complete lattice, a monotone function $f$ has a nonempty complete sublattice of fixed points. When $f$ is continuous, we can obtain the least fixed point by iterating $f$ on the bottom element of the lattice. While this formulation has been applicable in numerous contexts, a limitation is that the iteration is a sequential process and computing the fixed point requires as many steps as the number of iterations. In this paper, we propose fixed point methods on lattices that enable {\em parallel} computation of the fixed point.
Modern parallel and distributed systems increasingly operate under models where:
(1) updates overwrite state rather than merge it (e.g., key-value stores, 
parameter servers in machine learning),
(2) different workers maintain different components of the global state, and 
(3) communication delays mean workers operate on slightly stale data.
Classical fixed-point theorems do not directly address this setting.
This paper models a parallel or a distributed computation on $n$ cores (or $n$ machines) as computation of the
least {\em common} fixed point of $n$ functions.
We apply our method to many problems, such as the stable marriage problem and the shortest paths problem.

We begin with a fixed-point theorem with multiple monotone functions. Assuming interleaving semantics of
parallelism, the theorem gives a method to compute the least common fixed point. It assumes that
if there are two functions $f_1$ and $f_2$, then when they are computed in parallel, the net effect is
$f_1$ followed by $f_2$ or vice versa. This requires that the programmer uses appropriate synchronization 
for the interleaving guarantee. This result requires the functions to be inflationary, and the system uses
{\em fairness} in scheduling of the functions.

Then, we study when these functions are applied in the non--interleaving semantics of parallelism.
We add two constraints. First, instead of any lattice, we require the lattice to be distributive.
Many distributed (and some parallel) computations that admit consistent-cuts can be viewed as a poset of width $n$.
It is well known that the set of {\em ideals} of any finite poset forms a finite distributive lattice and
conversely every finite distributive lattice can be generated in this manner \cite{Birk3}.
In the context of parallel and distributed systems, we can rephrase this observation as the
set of consistent global states of a parallel or a distributed computation is a distributive lattice \cite{Matt:VTime}.
Second, we assume that a processor never issues a write for a coordinate that has not changed.
We call this assumption {\em update-only-on-change}. We show that the system converges to the least fixed point under these conditions.


Finally, we study when these functions are in a distributed system where the global state $G$ is distributed across multiple machines and any machine $i$ may have an older value of $G[j]$ for $j \neq i$. We now require the functions to satisfy an additional assumption called {\em bounded staleness}.
A distributed system with $n$ processes $P_1, P_2, \ldots, P_n$ maintaining the global variable
$G$ such that $G[i]$ is at $P_i$ satisfies {\em bounded staleness} if whenever $P_i$ reads
$G$ it gets the most recent value of $G[i]$ and a value of $G[j]$ that is at most $T$ rounds old.
We also require that functions satisfy an additional assumption called {\em locality}. 
If the functions are {\em local}, then the
system reaches the least common fixed point in spite of the distributed execution. 
A function $f_i$ is $i$-local if for all $G$ in the lattice $L$, $f_i(G)[j]$ is equal to $G[j]$ for $j \neq i$. Thus,
the function $f_i$ can read the value of $G$ but can change only $G[i]$.

\begin{table}[t]
\centering
\small
\renewcommand{\arraystretch}{1.2}
\begin{tabular}{|l|c|c|c|c|}
\hline
\textbf{Result} 
& \textbf{Lattice} 
& \textbf{Functions} 
& \textbf{Requirements}  
& \textbf{Scheduling} 
\\ \hline

Knaster--Tarski~\cite{tarski1955}
& Complete
& Single, monotone
& \makecell{None}
& N/A
\\ \hline

Kleene~\cite{kleene1952}
& Complete
& Single, $\omega$-continuous
& \makecell{None}
& Sequential
\\ \hline

Cousot--Cousot~\cite{cousot1977sigplan}
& Complete
& Multiple, monotone
& \makecell{Inflationary\\Join-based merge%
\tablefootnote{%
Updates are combined using the lattice join operator, which guarantees monotonic
global progress even under chaotic asynchronous execution.}}
& Chaotic (unordered) \tablefootnote{the applications of update functions in an arbitrary order, but updates are
combined using lattice joins}

\\ \hline

\makecell{Baudet~\cite{baudet1978},\\
Bertsekas~\cite{bertsekas1983}}
& Complete
& Multiple, monotone
& \makecell{Inflationary\\Asynchronous reads%
\tablefootnote{%
The operators act on the \emph{entire} global state and compute fixed points of a
single composite monotone mapping; locality or coordinatewise overwrite semantics
are not assumed.}}
& Fair asynchronous \tablefootnote{Overlapping updates with stale reads and overwrite semantics}
\\ \hline

Theorem~\ref{thm:interleaved} (Interleaved)
& Finite
& Multiple, monotone
& \makecell{Inflationary}
& Fair interleaving
\\ \hline

Theorem~\ref{thm:parallel} (Parallel)
& Finite Product 
& Multiple, monotone
& \makecell{Inflationary\\update-only-on-change}
& Fair parallel 
\\ \hline

Theorem~\ref{thm:distributed} (Distributed)
& Finite Product 
& Multiple, monotone
& \makecell{Inflationary\\$i$-local\\Bounded staleness}
& Strongly fair 
\\ \hline
\end{tabular}
\caption{Comparison of assumptions required by classical fixed-point theorems and the results of this paper.}
\label{tab:assumptions}
\end{table}

Table~\ref{tab:assumptions} highlights the structural and computational assumptions
underlying classical fixed-point theorems and contrasts them with the assumptions
required by our results.
Knaster--Tarski~\cite{knaster1928,tarski:fp} and Kleene~\cite{kleene1952} address the existence and computation of fixed points for
a \emph{single} monotone operator on a complete lattice, without reference to
parallel or distributed computing concerns.
Cousot and Cousot~\cite{cousot1977sigplan} consider multiple monotone functions on the lattice.
They guarantee convergence by enforcing idempotent join-based merge.
We guarantee convergence despite destructive overwrites.
The results of Baudet~\cite{baudet1978} and Bertsekas~\cite{bertsekas1983} consider \emph{multiple} monotone operators but rely
on contractions to ensure convergence under asynchrony. We do not assume that our functions are contractions.
We also refer the interested reader to \cite{FrommerSzyld2000} and references therein for a survey of such
asynchronous iterations.

Our theorems address a different and increasingly common computational model in which
updates overwrite state rather than merge it, and where locality, parallelism, and
distributed execution play a central role.
The table makes explicit how finiteness, inflationarity, coordinate locality, fairness,
and bounded staleness interact to guarantee convergence to the least common fixed point.


The term ``asynchronous'' is used with different meanings in the literature.
In abstract interpretation, chaotic iteration~\cite{cousot1977sigplan} refers to
the application of functions in an arbitrary order, but updates are
combined using lattice joins, ensuring a well-defined monotone global state.
In contrast, asynchronous scheduling in the sense of Baudet and Bertsekas allows
overlapping updates with stale reads and overwrite semantics, and generally cannot
be reduced to any sequential interleaving.
Our model departs from both by allowing overwrite-based parallel updates while
enforcing locality and bounded staleness to retain convergence guarantees.
To our knowledge, this is the first exact least-fixed-point convergence result 
for coordinate-wise overwrite semantics
that does not rely on join-based merging or contraction assumptions.

\section{Fixed Points on Lattices Under Interleaving Execution}\label{sec:interleaving}

Let $L \subseteq \mathbb{N}^n$ be a finite lattice ordered with the bottom element $\bot$.
Let ${\cal F}$  be a set of $m$ functions from $L$ to $L$.
We study the convergence of computation when $\{f_i \in {\cal F} | 1 \leq i \leq m\}$ are applied 
to some element $x\in L$ one at a time in some order. We assume that all the functions 
in ${\cal F}$ are inflationary and monotone.

\begin{definition}[Inflationary]
  A function $f:L\to L$ is inflationary if 
  \[
  \forall G\in L,\ f(G)\ge G.
  \]  
\end{definition}

\begin{definition}[Monotone]
    A function $f:L\to L$ is  monotone if \[\forall G, H \in L: G \le H \Rightarrow f(G) \le f(H).\]
\end{definition}

We define the {\em least common fixed point} of ${\cal F}$ as
\[ G^* \;=\; \min \{\, G \in L \mid \forall i,\; f_i(G)=G \,\}  \]
Observe that the set of fixed points for each of the functions $f_i$ is a sublattice of $L$.
Hence, the set of common fixed points is an intersection of the sublattices and therefore 
also a sublattice $Z$. Since this sublattice $Z$ is nonempty (the $\top$ element belongs to each sublattice because each function is inflationary), the operator $\min$ used above is well-defined.

Throughout, we will assume that the scheduler that applies the functions $f_i$ is {\em fair}, i.e., it selects which $f_i$ to apply, at each step infinitely often~\cite{francez2012fairness}. This leads to the following result.
\begin{theorem}[Convergence to the least common fixed point under a fair schedule]
\label{thm:interleaved}
Let $ {\cal F} = \{f_i : L \to L \mid i = 1,\dots,m\}$ be a set of inflationary and monotone
functions.
Consider an execution starting from $G_0=\bot$, where at each step $t$, a
fair scheduler selects an index $i \in [m]$ and performs the atomic update
$G_{t+1} = f_{i}(G_t)$. 
Then, the execution terminates at a state $G^*$ which is the least common fixed point 
of ${\cal F}$.
\end{theorem}

\begin{proof}
Since the schedule applies $G_{t+1}=f_{i}(G_t)$ and each $f_{i}$ is inflationary,
we have $G_t \le G_{t+1}$ for all $t$. Hence,
\[
G_0 \le G_1 \le G_2 \le \cdots
\]
is an ascending chain. Because $L$ is finite, the chain stabilizes: there exists
$T$ such that $G_T=G_{T+1}=\cdots$. Let $G^*:=G_T$.

We claim $G^*$ is a common fixed point. Suppose not; then for some $i$,
$f_i(G^*)\neq G^*$. Since the scheduler is fair, index $i$ is selected at some time
$t\ge T$, and then
\[
G_{t+1}=f_i(G_t)=f_i(G^*)\neq G^*,
\]
contradicting that the sequence is constant from time $T$ onward. Hence
$f_i(G^*)=G^*$ for all $i$.

Now let $Y$ be any common fixed point. We show by induction that $G_t \le Y$ for all $t$.
The base case holds since $G_0=\bot\le Y$. If $G_t\le Y$, then by monotonicity,
\[
G_{t+1}=f_{i}(G_t)\le f_{i}(Y)=Y.
\]
Thus $G^*\le Y$. Therefore, $G^*$ is the least common fixed point.
\end{proof}

\noindent
We now examine the assumptions of Theorem~\ref{thm:interleaved}:
(1) $L$ is finite, (2) The scheduler is fair, (3) $G_0$ is the bottom element of $L$, (4)
Each $f_i$ is monotone, and (5) Each $f_i$ is inflationary. We show that these assumptions are necessary in a weak sense: if any assumption is not met, there is a counterexample that results in a failed termination or non-existence of a least fixed point. The proof of Theorem~\ref{thm:assumptions} is deferred to the Appendix.


\begin{theorem}{\label{thm:assumptions}}
Consider executions of the form 
$G_{t+1} = f_{i}(G_t)$,
where at each step a scheduler selects an index $i \in [m]$ and applies the corresponding update
function $f_{i} : L \to L$. 
For each assumption $A$, there exists a lattice $L$,
a collection of functions $\{f_1,\dots,f_m\}$, an initial state $G_0$, and a schedule $S$ such that
all assumptions except $A$ are satisfied, yet the resulting execution either
(i) fails to terminate, or
(ii) terminates at a state that is not a common fixed point, or
(iii) terminates at a common fixed point that is not the least common fixed point.
\end{theorem}

\subsection{Transitive Closure Example}
As a simple application of Theorem~\ref{thm:interleaved}, consider the problem of computing
the transitive closure of a directed graph.
Let $V=\{1,\dots,n\}$ be the set of vertices.
For each ordered pair $(a,b)\in V\times V$, introduce a Boolean variable
$R[a,b]\in\{0,1\}$ indicating whether $b$ is reachable from $a$.
The global state is 
$R \in L := \{0,1\}^{n^2}$,
ordered componentwise. $L$ is a finite product lattice with bottom element
$\bot$.

The initial state $R_0$ is defined by
\[
R_0[a,b]=1 \quad\text{iff}\quad (a,b)\in E \text{ or } a=b.
\]

For each pair $(a,b)\in V^2$, define a function
$f_{a,b}:L\to L$ by
\[
f_{a,b}(R)[x,y] :=
\begin{cases}
R[a,b] \,\lor\, \bigl(\exists k: R[a,k]\land R[k,b]\bigr),
& \text{if } (x,y)=(a,b),\\[1mm]
R[x,y], & \text{otherwise}.
\end{cases}
\]




Each $f_{a,b}$ is monotone and inflationary: it only changes coordinate
$(a,b)$, and does so by replacing $R[a,b]$ with a disjunctive expression that can only
turn a $0$ into a $1$.

Consider an execution starting from $R_0$ in which, at each step, a fair
scheduler selects a pair $(a,b)$ and applies the update
$R_{t+1} = f_{a,b}(R_t)$.
By Theorem~\ref{thm:interleaved}, this execution terminates at the least common
fixed point of the family $\{f_{a,b}\}$.
At the fixed point, the condition
\[
R[a,b] = R[a,b] \lor (R[a,k]\land R[k,b]) \quad \text{for all } k,a,b
\]
holds, which is equivalent to transitivity of reachability.
Therefore, the least common fixed point reached by the execution is exactly
the transitive closure of the graph.







\section{Fixed Point Theorem on Finite Distributive Lattices under Non--Interleaving Parallelism}


We now consider the application of multiple $f_i$ functions in parallel without the interleaving semantics.
Theorem~\ref{thm:interleaved} assumes that at each step a \emph{single} function
$f_i$ is applied atomically, i.e., $G_{t+1}=f_i(G_t)$.
If instead multiple processes apply different functions \emph{in parallel} using
a naive ``read--compute--write'' rule, the computation can fail: the execution may
not terminate at a common fixed point (even when all $f_i$ are monotone and
inflationary).
In one parallel round, each process $i$ reads the same old state $G$,
computes $f_i(G)$, and then all processes write their results concurrently.
In particular, a later write may overwrite earlier writes with stale data.
Note that each process writes all coordinates of its computed vector, not just its own.
Thus, if writes are not coordinated, the resulting state can be an arbitrary
interleaving of coordinate-wise overwrites.
Equivalently, the last-writer-wins on each coordinate; the outcome is a nondeterministic
interleaving of writes.
To see this, let $L=\{0,1\}^2$ with $\bot=(0,0)$. Define two functions $f_1,f_2:L\to L$ by
\[
f_1(x_1,x_2)=(1,x_2),\qquad f_2(x_1,x_2)=(x_1,1).
\]
Both $f_1$ and $f_2$ are monotone and inflationary.

Starting from $G_0=(0,0)$, consider one parallel round where process~1 applies $f_1$
and process~2 applies $f_2$ to the \emph{same} old state $G_0$:
\[
f_1(G_0)=(1,0),\qquad f_2(G_0)=(0,1).
\]
Under the naive parallel write semantics, it is possible that the final written
state is
\[
G_1=(0,0)
\]
(e.g., process~1 writes coordinate~1 first, then process~2 overwrites coordinate~1
with $0$ from its computed state, and similarly for coordinate~2).
Thus, the system can remain stuck at $(0,0)$ forever:
\[
(0,0)\ \Rightarrow\ (0,0)\ \Rightarrow\ (0,0)\ \Rightarrow\ \cdots
\]

However, $(0,0)$ is \emph{not} a common fixed point:
\[
f_1(0,0)=(1,0)\neq(0,0),\qquad f_2(0,0)=(0,1)\neq(0,0).
\]
Hence, in this parallel execution the computation does not converge to a common
fixed point, contradicting the conclusion of Theorem~\ref{thm:interleaved}.
The counterexample above shows that the interleaving convergence guarantee does not extend
to naive parallel writes.
This behavior is prevented by the {\em update-only-on-change} assumption used in Theorem \ref{thm:parallel}.


Observe that the functions overwrite the vector with the last-writer-wins per coordinate.
If we naively allow both processes to write to all 
coordinates with last-writer-wins semantics—then convergence can fail.
We had assumed that a processor writes the value
$H[i]$ {\em even if $G[i]$ equals $H[i]$}. It was this reason that under non-interleaving semantics, $G_1$'s both components got the value $0$. 
We now assume that a function $f_i$ performs a write for component $i$ only if that
$i$'s component value has changed.
We introduce the  {\em update-only-on-change} model.
\begin{definition}[Update-only-on-change]
A processor satisfies update-only-on-change model if it issues an update on the component
$G[i]$ only if it has changed from what it read.
Equivalently, the execution of $f$ on $G$ is implemented as follows:
Let $H = f(G)$. 
For all $i$, if $(G[i] \neq H[i])$ issue a write on the $i^{th}$ component.
\end{definition}
From an implementation perspective, this assumes that any $f$ has access to the old value of its argument. When the new value is changed it can check if the value has changed and issues writes
only on the changed values.
We are assuming that the lattice $L$ is finite and each component of $G$ is simply an integer.
Thus, we do not have to worry about checking equality for two floating-point numbers.





We define \emph{non--interleaving semantics} as follows.
Execution proceeds in sequential rounds $t=0,1,2,\dots$ with committed states $G_t$.
In round $t$, a set of functions $\{f_i : i\in S_t\}$ is executed in parallel.
Each function $f_i$ may read coordinates of the shared state at different times during the
round; hence some coordinates may be read
before other functions' writes take effect.

The writes of the functions may overlap and can take effect in any order; if multiple
functions write to the same coordinate, the final value at the end of the round is determined
by the last write. The round completes once all started functions finish and their writes
have taken effect, yielding the next committed state $G_{t+1}$.
In this model, when a {\em read} overlaps with any number of {\em writes} it can return the value prior to all
the writes or after any of the writes. Similarly, when multiple writes overlap, then the final
value returned could be any of the {\em writes}. In particular, we do not require memory consistency
guarantees such as {\em sequential consistency} \cite{Lamport79a} or linearizability \cite{Herlihy-Wing90}.
The {\em update-only-on-change} model assumes that a function writes to coordinate $k$ if and only if its computed value for coordinate $k$ differs from the value it read for coordinate $k$.

A non--interleaving execution can be viewed as a computation where
each round starts from a committed global state, and processes may observe partial
effects of other writes within the same round, but no process ever observes a
state older than the round's starting state.
Although writes may overwrite each other, update-only-on-change ensures that no write can decrease a coordinate. 

We also assume that our lattices are now distributive. If the parallel or distributed computation is modeled as a finite poset of width $n$, then the set of consistent
global states forms a finite distributive lattice \cite{Birk3}.
Let the poset $P$ with $n$ chains be written as $P_1, P_2, \ldots, P_n$.
If the lattice $L$ is generated from $P$,
any element $G \in L$ can be viewed with $G[i]$ as the number of elements
in $P_i$. Since we will be using this representation of the lattice, distributivity is assumed in the rest of the paper. For simplicity, we will assume that our lattice is simply a finite lattice of vectors of natural numbers such that each component $G[i]$ of the vector is a natural number between $0$ and $m$.

\begin{theorem}[Parallel convergence with non-interleaving semantics]
\label{thm:parallel}
Let $L$ be a finite lattice of vectors $\mathbf{N}_m^n$ ordered with the usual bottom element $\bot = (0,0, \ldots, 0)$ with $\leq$ defined as the component-wise comparison.
For each $i\in\{1,\dots,n\}$, let $f_i:L\to L$ be monotone and inflationary.
Consider a parallel execution starting from $G_0=\bot$ in which, at each round $t$,
a scheduler selects a set $S_t\subseteq\{1,\dots,n\}$ and the next state is defined by
$G_{t+1} :=$ the global state reached after any parallel execution of $f_i$ on $G_t$ where $i \in S_t$.
Assume the scheduler is fair: each index $i$ belongs to $S_t$ for infinitely many rounds.
The execution assumes the update-only-on-change model.
Then, the execution terminates at a state $G^*$ which is the least common fixed point 
of ${\cal F}$.
\end{theorem}

Before we give the proof,  we first show the following Lemma.
\begin{lemma}[Monotone round progress]
\label{lem:round-progress}
Consider a round $t$ starting from the committed state $G_t$ under the
non--interleaving semantics.
Assume that each function $f_i$ is inflationary.
Assume further that updates are performed only when the computed value
differs from the value read.

Then for every coordinate $j$:
\begin{enumerate}
\item Every write to coordinate $j$ during round $t$ writes a value
      greater than or equal to $G_t[j]$.
\item The committed value $G_{t+1}[j]$ satisfies $G_{t+1}[j]\ge G_t[j]$.
\item Moreover, $G_{t+1}[j]=G_t[j]$ if and only if no function writes
      coordinate $j$ during round $t$.
\end{enumerate}
\end{lemma}

\begin{proof}(Lemma~\ref{lem:round-progress})
Fix a coordinate $j$.
Let a function $f_i$ read value $v$ for coordinate $j$ during round $t$.
Since we are in round $t$, $v \ge G_t[j]$.
Since $f_i$ is inflationary, the value $v'$ computed by $f_i$ for coordinate
$j$ satisfies $v' \ge v \ge G_t[j]$.
Thus every write to coordinate $j$ during the round writes a value at least
$G_t[j]$, proving (1).

Because the committed value $G_{t+1}[j]$ is determined by one of the writes
during the round (or remains unchanged if there are no writes),
it follows that $G_{t+1}[j]\ge G_t[j]$, proving (2).

Finally, if $G_{t+1}[j]=G_t[j]$ and some function had written coordinate $j$,
then by the update--only--on--change rule that write would have been strictly
greater than the value read, and hence strictly greater than $G_t[j]$,
contradicting $G_{t+1}[j]=G_t[j]$.
Conversely, if no function writes coordinate $j$ during the round,
then $G_{t+1}[j]=G_t[j]$ by definition of the execution model.
\end{proof}

\begin{proof} (Theorem \ref{thm:parallel})
We first show that the state sequence is monotone and stabilizes.
Fix a round $t$ and let $S_t$ be the set selected by the scheduler.
From Lemma \ref{lem:round-progress}, we get that $G_{t+1} \ge G_t$ componentwise.
Therefore, the sequence
\[
G_0 \le G_1 \le G_2 \le \cdots
\]
is an ascending chain in the finite lattice $L$. Hence, it stabilizes:
there exists $T$ such that $G_T=G_{T+1}=\cdots$. Let $G^*:=G_T$.

We now show that $G^*$ is a common fixed point.
We first show that for every $i$, $f_i(G^*)=G^*$. Suppose not. Then $f_i(G^*)[j]\ne G^*[j]$ for some $j \in [n].$
Because $f_i$ is inflationary,
this implies $f_i(G^*)[j]>G^*[j]$.
By fairness, there exists a round $t\ge T$ with $i\in S_t$. Since $G_t=G^*$ for all
$t\ge T$, the update rule gives
\[
G_{t+1}[j] > G_t[j],
\]
If any write to coordinate $j$ occurs during the round $t$, Lemma \ref{lem:round-progress} guarantees
that the committed value $G_{t+1}[j] > G_t[j]$ regardless of which functions's write wins.

This case arises when we have overlapping writes 
and a write by some other function may have won. 
Although overlapping writes may cause the committed value $G_{t+1}[j]$ to differ
from $f_i(G^*)[j]$, Lemma~\ref{lem:round-progress} guarantees that any committed
write to coordinate $j$ during round $t$ is strictly greater than $G_t[j]$.
Hence $G_{t+1}[j] > G_t[j]$, contradicting the assumption that the execution has
stabilized at $G^*$.
Thus, $G^*$ is a common fixed point.

We now show that $G^*$ is the least common fixed point.
Let $Y$ be any common fixed point, i.e., $f_i(Y)=Y$ for all $i$.
We prove by induction that $G_t\le Y$ for all $t$.
The base case holds since $G_0=\bot\le Y$.
Now suppose that we are at round $t$ and the state is $G_t$.
Let $S_t$ be the set of functions applied in the round $t$
and $G_{t+1}$ be the global state reached.
Suppose a function $f_i$ in $S_t$ reads a vector $X$. Notice that if $G_t \leq Y$ and $Y$ is a common fixed point, then all writes in round $t$ are $\leq Y$ and so the read vector $X \leq Y$.
Then by monotonicity, for all $f_i$
\[
f_i(X) \le f_i(Y)=Y,
\]
Therefore, the committed last-writer-wins state at end of round also satisfies $G_{t+1} \le Y$.
Hence, we obtain that  $G^*\le Y$.
Thus, $G^*$ is the least common fixed point of $\mathcal{F}$.
\end{proof}

\subsection{Transitive Closure Example}
Continuing with our example on  transitive closure, we now analyze the behavior on a parallel computer in the 
non-interleaving model.
From Theorem \ref{thm:parallel}, we obtain that the following program computes the transitive closure of $R$.
The function is shown as Algorithm \ref{alg:transitive2}.

\begin{algorithm}[ht]
\label{alg:transitive2}
\caption{Function $f_{(a,b)}$ for reachability update: The function is not synchronized and multiple threads may be executing it concurrently.}
\DontPrintSemicolon

\KwIn{indices $a,b$}
\KwData{$R[1\ldots n,1\ldots n]$}

\lIf{$R[a,b] = 1$}{
    \Return
}

 // $k$ is a local variable
 
\For{$k \gets 1$ \KwTo $n$}
    {\If{$R[a,k] \wedge R[k,b]$}{
        $R[a,b] \gets 1$\;
        \Return
    }
}

\end{algorithm}


First, observe that the function is not synchronized. Thus, multiple threads may be running them concurrently
without any synchronization. The variable $R$ is shared among threads and the variable $k$ is local to the function. 
The program also satisfies update-only-on-change. Only function $f_{a,b}$ can write $1$ on $R[a,b]$. It does that only if the prior value was $0$.
From Theorem \ref{thm:parallel}, we obtain that the program will yield the transitive closure of the
matrix $R$.
If all $n^2$ threads run in every round, then the algorithm will reach the least fixed point in
$O(\log n)$ rounds because in round $i$, $R[a,b]$ equal to $1$ if there is a path of length 
$2^i$. Note that the parallel round complexity holds only when every thread runs in every round.

The reachability example used only reads and writes.
As shown by Herlihy \cite{Herlihy:1988a}, we need objects with higher consensus number when 
we build objects under lock-free synchronization. We will use the compare-and-set operation on
the components of $G$ whenever necessary. For example, consider the simple task of finding the number of components in an array $A$ of size $m$ that are greater than $c$. We define a lattice $L = \{0,1\}^m \times \{0,\ldots, m\}$.
The first $m$ components indicate which entries from $1..m$ have been checked and 
the last component indicates the number of entries that are greater than $c$.
The lattice $L$ is ordered component wise.
Algorithm \ref{alg:useOfCAS} gives the code for $f_i$ for $i \in [1\ldots m]$.
In this example, if compare-and-set (CAS)\footnote{We use $CAS(a,b,c)$ to denote a single atomic instruction that checks if $a = b$ and if so, then $CAS$ returns true and sets $a$  to $c$. Otherwise, $CAS$ returns false.} fails for $f_i$, then $G.checked[i]$ is false, and is checked again
in a future round.
Note that although threads run concurrently, there is at most one active execution of $f_i$.

\begin{algorithm}[ht]
\label{alg:useOfCAS}
\caption{The function $f_i$ is not synchronized and multiple threads may be executing it concurrently.}
\DontPrintSemicolon

\KwIn{index $i$}
\KwData{G: struct \{ $checked$: array[$1..m$] of Atomic Boolean initially false;\\ \hspace*{1.2in} $sum$: Atomic Integer initially $0$;\} }
\lIf {$G.checked[i]$}{\Return}
\lIf{$(A[i] \leq c)$}{ \{ $G.checked[i] \gets true;$ \Return \} }
int $temp \gets atomicRead(G.sum)$\\
\lIf{$CAS(G.sum, temp, temp+1)$}{
     $G.checked[i] \gets true$
}
// if $CAS$ returns false, then the function $f_i$ is computed again\\
\end{algorithm}

\section{Fixed Point Theorem on Finite Distributive Lattices under Distributed Computation}

We now explore the program under the distributed computation model.
Suppose that we are working in a distributed computing model where $G$ is maintained in a distributed fashion.
We will assume that there are $n$ functions and the process $i$ applies the function $f_i$.
Thus, process $P_i$ keeps the most up-to-date copy of $G[i]$. 
We assume that when process $i$ applies the function $f_i$ it changes only the component $i$ that is locally stored.
It can then relay this value to other processes by sending messages.
\begin{definition}[$i$-local function]
\label{def:locality}
A function $f_i : L \to L$ is said to be \emph{$i$-local} if for all $G \in L$:
 $$
    \forall j \neq i:\quad f_i(G)[j] = G[j].
    $$
\end{definition}
Thus, if $f_i$ is $i$-local, it may update only the $i^{th}$ coordinate and leaves all other
coordinates unchanged (i.e., it does not write the value of other components).
The function $f_i$ may read the states from other processes, but can update only its own component.

When $P_i$ applies $f_i$, it has the most recent $G[i]$ but other components of $G$ may be old. Different processes may have different old values of the components for
other processes. So, Theorem \ref{thm:parallel} may not hold.
Let $L=\{0,1\}^3$, and $\bot=(0,0,0)$
ordered componentwise. Define three functions $f_1,f_2,f_3:L\to L$ by
$f_1(x_1,x_2,x_3) := (1,\,x_2,\,x_3),
f_2(x_1,x_2,x_3) := (x_1,\, x_1\vee x_2,\, x_3),
f_3(x_1,x_2,x_3) := (x_1,\,x_2,\, x_2\vee x_3).$

Then, each $f_i$ is monotone, inflationary, and $i$-local.
Consider a distributed model with three processes, where $P_i$ stores $G[i]$.
When $P_i$ executes $f_i$, it reads its own coordinate $G[i]$ fresh but may read
the other two coordinates stale. 

Start from $G_0=(0,0,0)$ and run a fair schedule that repeatedly executes all three
functions. We exhibit an execution that stabilizes at a state that is \emph{not} a common fixed point.
$P_1$ executes $f_1$ and sets $G[1]$ to $1$.
Suppose that due to asynchronous communication,
$P_2$ and $P_3$ continue to read $G[1]=0$ for an arbitrarily long time,
even though $G[1]=1$ at $P_1$. 
Whenever $P_2$ runs $f_2$, it reads its own coordinate $G[2]=0$
fresh but reads $G[1]=0$ stale. Hence, it computes
\[
f_2(0,0,0)[2] = 0\vee 0 = 0,
\]
and so it never changes $G[2]$.
Similarly, $P_3$ never changes $G[3]$.
Thus, under this fair schedule of local computations, the system stabilizes at
$G^*=(1,0,0)$. However, $G^*$ is \emph{not} a common fixed point, because
\[
f_2(G^*) = f_2(1,0,0) = (1,1,0)\neq (1,0,0).
\]
Therefore, the conclusion of Theorem~\ref{thm:parallel} fails in this distributed
stale-read model.

To restore the convergence to the least fixed-point, we add the {\em bounded staleness} assumption and strengthen the notion of fairness.

\begin{definition}[Bounded Staleness]
  A distributed system with $n$ processes $P_1, P_2, \ldots, P_n$ maintaining the global variable
  $G$ such that $G[i]$ is at $P_i$ satisfies bounded staleness with parameter $T$ if whenever $P_i$ reads
  $G$ it gets the most recent value of $G[i]$ and a value of $G[j]$ that is at most $T$ rounds old.
 
\end{definition}

\begin{definition}[Strong Fairness]  
A scheduler is strongly fair if every process is selected at least once in 
any window of $\tau$ consecutive rounds.
\end{definition}

Observe that the strong fairness of the schedule does not imply bounded staleness.
Strong fairness constrains which processes execute; bounded staleness constrains what information they observe.
Strong fairness and bounded staleness are complementary: strong fairness 
with parameter $\tau$ ensures every process executes within $\tau$ rounds, while 
bounded staleness with parameter $T$ ensures every process observes data at 
most $T$ rounds old. Together, these conditions guarantee that after 
stabilization, every process eventually operates on the true fixed point.

Since computing $f_i$ requires values from other processes, we keep the view of the $G$ vector
at each of the processes. We denote this view by $\widehat{G}^{(i)}$.
However, from the perspective of the computation, we consider the
vector $G$ with $i^{th}$ component at $P_i$. Thus, $G$ is a distributed vector.

\begin{theorem}[Distributed least fixed point theorem]
\label{thm:distributed}
Let $L$ be a finite lattice of vectors $\mathbf{N}_m^n$.
Let $\{f_i:L\to L\}_{i=1}^n$ be a set of monotone, inflationary, and $i$-local functions.
Consider an execution $G_0=\bot, G_1,G_2,\dots$ generated as follows.
At each round $t$, a strongly fair scheduler chooses a set $S_t\subseteq [n]$, and for each
$i\in S_t$ process $i$ computes its update using its view $\widehat G^{(i)}_t\in L$ satisfying bounded staleness.
The commit rule is coordinatewise:
\[
G_{t+1}[i]:= f_i(\widehat G^{(i)}_t)[i]\ \ (i\in S_t),
\qquad
G_{t+1}[j]:=G_t[j]\ \ (j\notin S_t).
\]

Then the execution terminates at the least common fixed point $G^*\in L$.
\end{theorem}

\begin{proof}
Fix round $t$. For $i\in S_t$, the inflationary property gives $f_i(\widehat G^{(i)}_t)\ge \widehat G^{(i)}_t$,
hence
\[
G_{t+1}[i]=f_i(\widehat G^{(i)}_t)[i]\ \ge\ \widehat G^{(i)}_t[i]=G_t[i].
\]
For $j\notin S_t$, $G_{t+1}[j]=G_t[j]$. Thus, $G_{t+1}\ge G_t$ componentwise.
Since $L$ is finite, the ascending chain $G_0\le G_1\le\cdots$ stabilizes:
there exists $U$ such that $G_U=G_{U+1}=\cdots$. Let $G^*:=G_U$.

Fix any index $i$. By strong fairness, there exists $t\le U+\tau$ with $i\in S_t$.
Because the global state has been constant from time $U$ onward, every coordinate
has the same value at times $t, t+1,\dots,t+\tau$. Therefore, bounded staleness implies
$\widehat G^{(i)}_t = G^*$.
Hence,
\[
G_{t+1}[i]= f_i(\widehat G^{(i)}_t)[i]= f_i(G^*)[i].
\]
But $G_{t+1}=G_t=G^*$, so $f_i(G^*)[i]=G^*[i]$. By $i$-locality, $f_i(G^*)=G^*$.
Since $i$ is arbitrary, $G^*$ is a common fixed point.

We now show that $G^*$ is the least common fixed point.
Let $Y$ be any common fixed point. We show by induction that $G_t\le Y$ for all $t$. Notice $G_0=\bot\le Y$. Assume $G_t\le Y$.
For $i\in S_t$, bounded staleness implies that each coordinate of $\widehat G^{(i)}_t$ is
from some earlier $G_{t-T}$, and since the sequence is ascending, $\widehat G^{(i)}_t\le G_t\le Y$.
By monotonicity, $f_i(\widehat G^{(i)}_t)\le f_i(Y)=Y$, hence
$G_{t+1}[i]=f_i(\widehat G^{(i)}_t)[i]\le Y[i]$. Coordinates not in $S_t$ are unchanged,
so $G_{t+1}\le Y$. Thus $G^*\le Y$, proving the claim.
\end{proof}

In Theorem \ref{thm:distributed}, we observe that the system may have terminated (i.e., reached the least fixed point) but none of the processes may be aware of it. Hence, when we
use the distributed version, we assume a termination detection algorithm that 
checks whether the computation has terminated. Any termination detection algorithm such as 
\cite{DijkScho:Term} is sufficient.

\subsection{Transitive Closure Example}
Now, we continue with our example of transitive closure.
The distributed code for process $P_{a,b}$ is shown in Algorithm \ref{alg:transitive3}. For simplicity, we maintain the entire matrix at each process, even though it is sufficient to maintain the vectors $\widehat{G}[a,*]$ and $\widehat{G}[*,b]$.
It is easy to verify that our computation is monotone, inflationary and {\em local}.
Assuming bounded staleness and strong fairness, 
from Theorem \ref{thm:distributed}, the algorithm terminates with the transitive closure of $R$. 

\begin{algorithm}[H]
\label{alg:transitive3}
\caption{Update for reachability coordinate $(a,b)$}
\DontPrintSemicolon

\KwData{local copy $\widehat{G}$ initially $R[1\ldots n,1\ldots n]$}
\KwIn{indices $a,b$}

\lIf{$\widehat{G}[a,b] = 1$}{
    \Return
}

\For{$k \gets 1$ \KwTo $n$}{
    \If{$\widehat{G}[a,k] \wedge \widehat{G}[k,b]$}{
        $\widehat{G}[a,b] \gets 1$\;
        \textbf{send} $\widehat{G}[a,b]$ to all processes except $P_{a,b}$\;
        \textbf{break}\;
    }
}

\medskip
\textbf{upon receive} $(value)$ \textbf{from} $P_{i,j}$:\;
\Indp
$\widehat{G}[i,j] \gets value$\;
\Indm

\end{algorithm}

Table \ref{tab:assumption-necessity} summarizes the necessity of the assumptions made by Theorems \ref{thm:interleaved}, \ref{thm:parallel}, and \ref{thm:distributed}.

\begin{table}[ht]
\centering
\small
\renewcommand{\arraystretch}{1.25}
\begin{tabular}{|l|p{9cm}|}
\hline
\textbf{Assumption} & \textbf{Failure when the assumption is removed} \\
\hline

Finite lattice
& The computation may not terminate.  
On an infinite chain, an inflationary function can strictly increase
forever without reaching a fixed point. \\
\hline

Inflationary updates
& The execution may oscillate and fail to converge.
Non-inflationary monotone functions can repeatedly undo progress. \\
\hline

Monotonicity
& The execution may converge to a non-least fixed point or overshoot the least
fixed point depending on the schedule. \\
\hline

Fair scheduling
& The execution may stabilize at a state that is not a common fixed point if some
coordinates are never updated. \\
\hline

Update-only-on-change
& Concurrent writes can overwrite progress with stale values, causing
lost updates and incorrect convergence. \\
\hline

$i$-locality
& Parallel executions with overwrite semantics can destroy progress made by
other functions, even under fairness and inflationarity. \\
\hline

Bounded staleness
& With unbounded stale reads, the system may converge to a spurious fixed point
that is strictly below the least common fixed point. \\
\hline

Product lattice
& Without product structure, coordinate-wise updates may not 
preserve lattice structure, leading to undefined states. \\
\hline

\end{tabular}
\caption{Necessity of assumptions: removing any assumption leads to failure of
termination, correctness, or least-fixed-point convergence.}
\label{tab:assumption-necessity}
\end{table}

\section{Greatest Common Fixed Points}

Analogous results hold for greatest fixed points using deflationary functions starting from $\top$.

\begin{corollary}
\label{cor:deflationary-interleaving}
Let $L$ be a finite lattice $(X,\leq) $ with the bottom element $\bot_L = (0,0, \ldots, 0)$. Let $L'$ be a finite lattice $(X, \geq)$ with the bottom element $\bot_{L'} = (m,m,\dots,m)$. 
For each $i\in\{1,\dots,n\}$, let $f_i:L'\to L'$ be monotone and deflationary. Then, one can compute the greatest common fixed point $G^*$
of ${\cal F}$ on $L$ by computing the least fixed point of $F$ on $L'$. Thus, one can find the greatest common fixed point in the interleaving semantics model.
\end{corollary}

\pagebreak

\begin{corollary}
\label{cor:deflationary-pardist}
Let $L$ be a finite lattice $(\mathbf{N}_m^n, \leq)$  with the bottom element $\bot_L = (0,0, \ldots, 0)$. Let $L'$ be a finite lattice $(\mathbf{N}_m^n, \geq)$ with the bottom element $\bot_{L'} = (m,m,\dots,m)$. 
For each $i\in\{1,\dots,n\}$, let $f_i:L'\to L'$ be monotone and deflationary. Then, one can compute the greatest common fixed point $G^*$
of ${\cal F}$ on $L$ by computing the least fixed point of $F$ on $L'$. Thus, one can find the greatest common fixed point in the parallel with update-only-on-change, and distributed with bounded staleness and strong fairness models.
\end{corollary}

\section{Applications}
In this section we apply Theorems~\ref{thm:interleaved}, \ref{thm:parallel}, and \ref{thm:distributed} to various fundamental algorithmic problems. 

\subsection{Stable Marriage}

As an example of the application of Theorem~\ref{thm:interleaved}, consider the implementation of the Gale-Shapley
algorithm \cite{gale1962college} for the stable marriage problem.
The stable marriage problem has $n$ men $\{m_1, \ldots, m_n\}$ and $n$ women $\{w_1, \ldots, w_n\}$. Each person has a strict preference ordering over members of the opposite gender. Let $pref[m][]$ denote the preference list of man $m$, and $rank[w][i]$ denote the rank of man $m_i$ in woman $w$'s preference list (a lower rank means more preferred).
For each man $m_i$, let $G[i] \in \{0, 1, \ldots, n+1\}$ denote his current proposal position in his preference list, where $G[i] = k$ means $m_i$ is currently proposing to his $k$-th choice (with $G[i] = 0$ meaning he has not yet proposed to anyone, and
$G[i] = n+1$ meaning that he has already proposed to all women and has been rejected by all.). The global state is
$$G \in L = \{0, 1, \ldots, n+1\}^n,$$
ordered componentwise.
In some implementations, we also keep $partner$ for each woman $w$. 

\begin{algorithm}[ht]
\label{alg:stable1}
\caption{Function $f(i)$ for proposal.
The function is synchronized.}
\DontPrintSemicolon

\KwIn{index $i$}
\KwData{$G[1 \ldots n]: 0 \ldots n+1$ // current proposal by man $i$ initially $0$}
\KwData{$partner[1 \ldots n]$ // current partner for woman initially $0$}
\KwData{$mpref[1 \ldots n][1 \ldots n]$ // read only: list of preferences of man $i$}
\KwData{$rank[1 \ldots n][1 \ldots n]$ // read only: ranking by women}

lock($G$); \\
\Indp

\lIf{$G[i] = 0$}{
    $G[i] \gets G[i]+1$
    \Return
}
\lIf{$G[i] > n$ }{
    \Return ~ // beyond the preference list
}

int $w \gets mpref[i][G[i]]$\\
\lIf{$(partner[w] = 0)$} 
    {$partner[w] \gets i$}
\Else{
   \lIf{$rank[w][i] < rank[w][partner[w]]$}
       {$partner[w] \gets i$}
    \lElse { $G[i] \gets G[i] + 1$}
    }
\Indm
unlock($G$);
\end{algorithm}

Algorithm~\ref{alg:stable1} gives a parallel interleaving implementation using locks.
The implementation is quite standard with locks being used to ensure correctness.
The algorithm will terminate with the man-optimal stable matching with the usual
arguments \cite{gale1962college}.

\begin{algorithm}[ht]
\label{alg:stable2}
\caption{Function $f(i)$ for proposal.
The function is not synchronized. The idea is an adaptation of Gale-Shapley's algorithm with non-interleaving parallelism.}
\DontPrintSemicolon


\lIf{$G[i] = 0$}{
    $G[i] \gets G[i]+1$
    \Return
}
\lIf{$G[i] > n$ }{
    \Return ~// beyond the preference list
}
int $w \gets mpref[i][G[i]]$\\
\For{int $j \gets 1$ \KwTo $n$, $j \neq i$}{ 
    \If{$(mpref[j][G[j]] = w) \wedge (rank[w][j] < rank[w][i])$}{ \label{code:G}
        $G[i] \gets G[i]+1$\;
        \Return
    }
}

\end{algorithm}

Now suppose that we want to solve the problem without using locks. 
Consider the implementation shown in Algorithm \ref{alg:stable2}. Notice that Algorithm \ref{alg:stable2} is inflationary and monotone.
The reader can verify that it works when the code is executed under a lock (i.e. with 
interleaving parallelism). We show that 
Algorithm \ref{alg:stable2} works even when it is not under any synchronization. In particular, while the condition at line \ref{code:G} is being checked, $G[j]$
may have changed. The variables $w$ and $k$ are local to the process $P_i$.
From Theorem \ref{thm:parallel}, it follows that computing $f_i$ in a non--interleaving fashion
will still yield the man-optimal stable marriage.
Algorithm \ref{alg:stable2} is correct under the non-interleaving round-based execution model of Theorem \ref{thm:parallel}, which assumes update-only-on-change.

Notice that Algorithm \ref{alg:stable2} is also $i-$local. Hence, we can also use its distributed version where $P_i$ maintains $G[i]$ assuming bounded staleness. Each process maintains a copy of the
$G$ vector, $\widehat{G}$ and sends $\widehat{G}[i]$ whenever it changes. It also receives
$\widehat{G}[j]$ for $j \neq i$ and updates it whenever a message is received. From Theorem \ref{thm:distributed}, we know that the system will converge to the man-optimal stable marriage.
As before, we can use a termination detection algorithm to detect when the system has converged.






We now discuss a slight variant of the stable marriage problem. 
Suppose we are interested in only that stable marriage in which man $m_1$ is married to $w_1$.
Observe that depending on $mpref$ and $rank$, there may not be any such stable marriage.
We have to make the following change to the program.
When man $m_1$ proposes to any woman other than $w_1$ or, if the woman $w_1$ is proposed by someone that she prefers to $m_1$, then  $m_1$ simply moves forward
(i.e. $G[1]$ is incremented). In the case there is no stable marriage, $P_1$ will move to $G[1]$ equal to $n+1$. 

\subsection{Shortest Paths}
Assume we have a weighted, directed graph $G = (V, E, w)$ on $n$ vertices with non-negative edge weights specified by the weight function $w$. We use $w[i,j]$ to denote the weight of edge $(i,j)$.
Consider a lattice $L$ of vectors of size $n$ where $G[i]$ denotes the length of the path from the vertex $v_1$ (source vertex) to $v_i$.
Let $T$ be the largest edge weight in the graph. This means that we have a path 
of weight at most $T(n-1)$ for any pair of vertices.
We initialize $G[i]$ to $w[1,i]$ when there is an edge from $v_1$ to $v_i$ and to
$nT$ when there is no edge from $v_1$ to $v_i$.
Notice that this lattice uses the componentwise $\geq$ operation instead of the usual $\leq$ operation.
The following simple edge-relaxation algorithm will converge to the length of the shortest path
from vertex $1$ to $i$ for every vertex $i$.
\begin{algorithm}[ht]
\label{alg:edge-relaxation}
\caption{ The function $f(i)$ for edge relaxation. It is not synchronized.}
\DontPrintSemicolon
\KwData{$G[1\ldots n]$ initially $G[i] = w[1,i]$;}
\lIf {$ \exists k: G[i] >  G[k] + w[k,i] $}{  $G[i] \gets G[k] + w[k,i] $}
\end{algorithm}

Algorithm \ref{alg:edge-relaxation} applies edge relaxation to compute shortest paths. 
Each function $f_i$ is deflationary, monotone, and $i$-local. 
By Corollary \ref{cor:deflationary-pardist}, the algorithm converges to the shortest path lengths. However, convergence may take exponential time in the worst case.

We now formulate the Bellman-Ford algorithm in the parallel non-interleaving model~\cite{CLRS2009}.
Let lattice $L$ be on the vectors $G$ of pairs, where $G[i].dist$ denotes the length of a path from 
vertex $v_1$ (source vertex) to $v_i$ and $G[i].level$ denotes the number of times the function $f_i$ has computed an edge-relaxation for $v_i$. 
Algorithm \ref{alg:Bellman-Ford} gives a parallel non-interleaving implementation of the Bellman-Ford algorithm.
The function $f_i$ checks if some other vertex $k$ is \emph{lagging}; the number of times an edge relaxation has been computed for $k$ is less than that of $i$. In that case, it simply returns and waits for the level number of that vertex to increase. Otherwise, it performs an edge-relaxation for all the incoming edges to $i$. The variable $pre(i)$ denotes the set of all vertices $k$ such that $(k,i)$ is an edge.
From Corollary \ref{cor:deflationary-pardist}, it follows that Algorithm \ref{alg:Bellman-Ford} converges to the
cost of the shortest paths to all vertices.

\begin{algorithm}[ht]
\label{alg:Bellman-Ford}
\caption{ The function $f(i)$ for Bellman-Ford Algorithm. It is not synchronized.}
\DontPrintSemicolon
\KwData{$G[1\ldots n]$ initially $G[i].dist = w[1,i]$; $G[i].level := 1$;}
\lIf {$\exists k: G[k].level < G[i].level $}
    { \Return ~ // some vertex lagging behind} 
\For {$k \in pre(i)$} {
    \lIf{$G[i].dist > G[k].dist + w[k,i]$ }{ $G[i].dist \gets G[k].dist + w[k, i]$} 
}
$G[i].level \gets G[i].level + 1;$
\end{algorithm}

Now consider the Floyd-Warshall Algorithm~\cite{CLRS2009}. 
Let lattice $L$ be of vectors of size $n \times n$ where $G[i,j]$ denotes the length of the path from the vertex $v_i$ to $v_j$.
We initialize $G[i,j]$ to $w[i,j]$ with $G[i,i]$ equal to $0$ and
$G[i,j]$ to $nT$ when there is no edge from $v_i$ to $v_j$. 

The non-interleaving function is specified as Algorithm \ref{alg:Floyd-Warshall}.\\
\begin{algorithm}[ht]
\label{alg:Floyd-Warshall}
\caption{$f(i,j)$ for Floyd-Warshall Algorithm.
It is not synchronized.}
\DontPrintSemicolon
\KwData{$G[1\ldots n,1\ldots n]$ initially $G[i,j] = w[i,j]$}
\lIf{$\exists k: G[i, j] > G[i, k] + G[k, j]$ }{ $G[i, j] \gets G[i, k] + G[k,j]$}
\end{algorithm}
Algorithm \ref{alg:Floyd-Warshall} implements Floyd-Warshall with non-interleaving semantics. The 
function $f(i,j)$ is deflationary, monotone, and $(i,j)$-local. By Corollary \ref{cor:deflationary-pardist}, the algorithm computes all-pairs shortest paths
correctly.

Finally, consider a graph with negative edge weights but no cycle with negative cost.
Algorithm \ref{alg:Johnson} computes the price vector $G$ such that 
if $w[i,j]$ is replaced with $w[i,j] + G[j] - G[i]$ then all edge weights are non-negative.
\begin{algorithm}[ht]
\label{alg:Johnson}
\DontPrintSemicolon
\KwData{$G[1\ldots n]$ initially $\forall i: G[i] = 0$ // $G$ is the price vector}
\lIf {$ \exists k: G[i] <  G[k] - w[i,k] $}{  $G[i] \gets G[k] - w[i,k] $}
\caption[]{The function  $f(i)$ for Johnson's Algorithm.
It is not synchronized.}
\end{algorithm}

These examples illustrate that many classic dynamic programming and graph algorithms can be viewed as monotone fixed-point computations on product lattices.
Observe that all the algorithms for the shortest path calculation satisfy $i$-locality.
Hence, assuming strong fairness and bounded staleness, they can also be implemented in a distributed system.
\subsection{Fair Division with Subsidy}

In this section, we show how to apply our results towards the fair division with subsidy problem~\cite{halpern2019fair}. 
We have a set $\mathcal{N}$ of $n$ agents, a set $\mathcal{M}$ of $m$ items and a pool of subsidy money. An allocation $X$ is an assignment of items to agents where $X_i$ is the set of items that $i$ receives.  Each agent $i$ has a valuation function $v_i$ and agent $i$'s value for agent $j$'s bundle is $v_i(X_j)$. An \textit{allocation with payments} $(X,p)$ is a tuple of an allocation $X$ and a payment vector $p = (p_1, \dots, p_n)$, where $p_i$ is the payment to agent $i$. Under such an allocation with payments $(X, p)$, agent $i$'s utility is $v_i(X_i) + p_i$. An allocation with payments $(X,p)$ is \emph{envy-free} and $p$ is \emph{envy-eliminating} if $v_{i}(X_i) + p_i \geq v_{i}(X_j) + p_j$ for all agents $i,j \in \mathcal{N}$. An allocation $X$ is \emph{envy-freeable} if there exists a payment vector $p$ such that $(X,p)$ is envy-free. 

\cite{halpern2019fair} prove that, given an envy-freeable allocation $X$, the set of all envy-eliminating vectors forms a distributive lattice under component wise maximum and minimum.

We will show how to apply Theorems~\ref{thm:interleaved}, \ref{thm:parallel}, and \ref{thm:distributed}, to compute an envy-eliminating payment vector given an envy-freeable allocation $X$. Let $\Delta$ be the largest $v_i(X_j)$ value for all $i,j \in \mathcal{N}$. The global state is:
\[
G \in L := \{0, 1, \dots, m\Delta\}^{n}
\] and corresponds to the subsidy vector $p$ (i.e $G[i] = p_i$). We start initially with 
$G = \{0\}^n$ and increase agents' subsidy until the resulting allocation with payments is envy-free. The function $f_i$ is given in Algorithm~\ref{alg:rental}. 

\begin{algorithm}[ht]
\label{alg:rental}
\caption{Function $f_i$ for computing an envy-eliminating subsidy vector.}
\DontPrintSemicolon
\KwData{$G[1\ldots n]$ initially $\forall i: G[i] = 0$ // $G$ is the subsidy vector}
int \textit{maxEnvy} $= 0$\;
\For{int $j \gets 1$ \KwTo $n$, $j \neq i$}{ 
    \If{$v_i(X_j) + G[j] - (v_i(X_i) + G[i]) >$ maxEnvy}{ 
        $maxEnvy \gets v_i(X_j) + G[j] - (v_i(X_i) + G[i])$\;
    }}
$G[i] \gets G[i] + maxEnvy$\\
\Return
\end{algorithm}

Notice that the function $f_i$ is \emph{inflationary} since we only increase agents' subsidy amount. We briefly show why each $f_i$ is also monotone. Notice that if the subsidy vector $G \leq H$, then we have that $\forall i: G[i] \leq H[i]$. When we apply any $f_i$ to $G$, we will add agent $i$'s maximum envy to its subsidy value, $G[i]$. Without loss of generality, let agent $j$ be the agent $i$ is most envious of. After applying $f_i$ to $G$, the only index that changes is $G[i]$ and it becomes $G[i] + (v_i(X_j) + G[j] - (v_i(X_i) + G[i]))$. Simplifying, we get, $(v_i(X_j) + G[j] - v_i(X_i))$ which is less than or equal to $(v_i(X_j) + H[j] - v_i(X_i))$ since $G \leq H$. This is exactly the subsidy payment to agent $i$ in $f_i(H)$. 

Theorems~\ref{thm:interleaved} and \ref{thm:parallel} result in parallel algorithms for computing envy-eliminating subsidy vectors given an envy-freeable allocation $X$. Since Algorithm~\ref{alg:rental} only modifies $i$'s subsidy payment, it is $i$-local and we can apply Theorem~\ref{thm:distributed} to get a distributed algorithm for computing envy-eliminating subsidy vectors assuming bounded staleness.

\section{Related Work}
Knaster-Tarski's fixed-point theorem states that any monotone function on a complete lattice
has a least fixed point, characterized as the meet of all pre-fixed points
(and dually for greatest fixed points) \cite{tarski1955}.
Kleene's theorem gives a constructive characterization of the least fixed point lfp$(F)$, 
\(
\mathrm{lfp}(F)=\bigvee_{k\ge 0}F^k(\bot)
\)
under additional continuity assumptions ($\omega$-continuity) \cite{kleene1952}.
Theorem~\ref{thm:interleaved} works with a \emph{family} of functions
$\{f_i\}$ and an arbitrary fair schedule, but it assumes (i) \emph{inflationary} updates
to ensure an ascending chain, and (ii) \emph{finiteness} of $L$ to guarantee
termination without requiring continuity. 

Cousot \& Cousot show that systems of monotone equations
over complete lattices can be solved by \emph{chaotic iterations}, where components are
updated in an arbitrary order, under suitable fairness
hypotheses \cite{cousot1977sigplan,cousot1977rr}.
These results are widely used to justify worklist algorithms for dataflow analysis.
Our theorems are different in multiple ways.
We assume our lattices to be finite and distributive. We do not use continuity of our functions.
Second, the standard chaotic-iteration setting typically models component updates that do not
lose information, conceptually a join-based merge, which is why parallel
overwrites can break convergence unless the merge is monotone. Since we are not using join-based merge, 
their results are not applicable in our setting. Also, we address computations where different processes
may have possible different views of the global state.
We observe that Conflict-free replicated data types (CRDTs)~\cite{shapiro2011crdt,shapiro2011comprehensive}
provide strong eventual consistency guarantees for distributed systems by
restricting updates to be inflationary over a join-semilattice and resolving
concurrent updates using a join operation.
As a result, replicas converge deterministically despite arbitrary message
reordering or duplication.
This model is closely related to join-based chaotic iteration in abstract
interpretation, but differs fundamentally from the overwrite semantics studied
in this paper. CRDTs thus fall under Cousot \& Cousot in Table~\ref{tab:assumptions}, as they rely on join-based merge semantics rather than overwrite-based updates.
In summary, our results do not imply results by Cousot \& Cousot, and their results do not imply our results.

Early results by Chazan and Miranker~\cite{chazan1969} and 
Baudet \cite{baudet1978} studied asynchronous iterative methods for multiprocessors, providing conditions for convergence of block-iterative fixed-point computations under
asynchrony.
Bertsekas analyzed distributed asynchronous computation of fixed points under models
with communication delays and partial asynchronism \cite{bertsekas1983}.
The monograph by Bertsekas and Tsitsiklis~\cite{bertsekas1989} provides a
comprehensive treatment of such models.
These works are typically formulated for iterative methods
$x^{t+1}=F(x^t)$  in analytic or order-theoretic settings, with
assumptions such as contraction or bounded delays. Our setting is purely
order-theoretic on a finite lattice: termination is guaranteed by the ascending-chain
stabilization, and fairness is the key scheduler assumption.

Bounded-staleness execution models arise in modern distributed and
parallel systems, particularly in large-scale machine learning.
The Stale Synchronous Parallel (SSP) model~\cite{cipar2013} and related systems
allow bounded stale reads while preserving convergence guarantees.
Although these works focus on optimization objectives rather than exact
fixed points, they provide practical motivation for our bounded-staleness
assumptions.
Modern graph processing systems such as Pregel \cite{Malewicz2010}, 
GraphLab \cite{Low2012}, and PowerGraph \cite{Gonzalez2012} use 
bulk-synchronous or asynchronous execution models for distributed graph 
computation. These systems often rely on commutative updates or eventual 
consistency rather than exact fixed-point convergence. Our theorems provide 
formal guarantees for a related but distinct model with overwrite semantics.
\section{Conclusions and Future Work}

We have presented three convergence theorems for computing least common 
fixed points under progressively relaxed synchronization: interleaving semantics
(Theorem \ref{thm:interleaved}), parallel with non-interleaving semantics (Theorem \ref{thm:parallel}), and distributed with 
bounded staleness (Theorem \ref{thm:distributed}). Our key insight is that {\em update-only-on-change} 
enables parallel convergence with overwrite semantics, while bounded 
staleness with {\em locality} extends this to distributed settings with stale reads. 
Applications include the transitive closure, stable marriage, shortest paths, and fair division with subsidy problems. 

Future work includes deriving convergence time bounds, extending to 
infinite lattices with appropriate continuity assumptions, and experimental 
evaluation on parallel and distributed systems. It would also be interesting to weaken the assumptions in our Theorems. 


\bibliography{fmaster}
\pagebreak
\appendix
\section{Appendix}

\subsection{Notation and Definitions}

\begin{table}[ht]
\centering
\small
\renewcommand{\arraystretch}{1.25}
\begin{tabular}{|l|p{9.5cm}|}
\hline
\textbf{Symbol} & \textbf{Meaning} \\
\hline

$L$ 
& A finite lattice; later specialized to a finite product lattice
$L=\prod_{i=1}^n L_i$ or a finite distributive lattice of ideals. \\
\hline

$\le$ 
& Partial order on $L$ (componentwise order in product lattices). \\
\hline

$\bot$, $\top$
& Bottom and top elements of the lattice $L$. \\
\hline

$G, G_t, G^*$ 
& Global state of the computation; $G_t$ is the state at round $t$,
$G^*$ is the terminal (fixed-point) state. \\
\hline

$G[i]$ 
& $i$-th coordinate of the global state vector $G$. \\
\hline

$f_i : L \to L$
& Update function associated with coordinate (or process) $i$. \\
\hline

$F$ 
& Set of update functions $\{f_1,\dots,f_n\}$. \\
\hline

$i$-local 
& Property of a function $f_i$ that it may update only coordinate $i$:
$f_i(G)[j]=G[j]$ for all $j\neq i$. \\
\hline

Inflationary
& Property of a function $f$: $f(G)\ge G$ for all $G\in L$. \\
\hline

Monotone
& Property of a function $f$: $G\le H \Rightarrow f(G)\le f(H)$. \\
\hline

$S_t$
& Set of indices whose functions are executed in parallel during round $t$. \\
\hline

Fair schedule
& Scheduling assumption that each index $i$ appears in $S_t$ infinitely often. \\
\hline

Strong fairness
& Stronger fairness: each index appears at least once in every bounded window
of rounds (used with bounded staleness). \\
\hline

Update-only-on-change
& Execution rule: a function writes to coordinate $j$ only if the value
of $j$ actually changes. \\
\hline

Bounded staleness ($T$)
& Assumption that reads of $G[j]$ return values at most $T$ rounds old. \\
\hline


$n$
& Number of coordinates / processes / functions. \\
\hline

$m$
& Upper bound on component values in vector lattices.\\
\hline

$R[a,b]$
& Boolean reachability variable in the transitive-closure application. \\
\hline

$a,b,k$
& Indices used in reachability updates and algorithmic examples
(reused locally). \\
\hline

$mpref, rank$
& Preference lists and ranking matrices in the stable marriage application. \\
\hline

partner
& Current tentative partner of a woman in the stable marriage. \\
\hline


\end{tabular}
\caption{Summary of notation used throughout the paper. Symbols marked as reused
appear in multiple sections with consistent meaning.}
\label{tab:notation}
\end{table}

\subsection{Proofs missing from Section~\ref{sec:interleaving}}

\begin{proof}[Proof of Theorem~\ref{thm:assumptions}]
    We analyze each assumption independently.
    
    We require $L$ to be finite; otherwise, even with inflationary and monotone functions, a fair scheduler $S$ and a starting point $G_0 = \bot$, convergence is not guaranteed and common fixed points may not exist. For example, let $L=\mathbb{N}$, $G_0 = 0$, with the usual order and define
$f(x)=x+1$. Since, we only have one function, any scheduler is fair. The function $f$ is inflationary and monotone, but it has no fixed
point.

Fairness is necessary to ensure convergence to a \emph{common} fixed point.
Let $m=2$ and $L=\{0,1\}^2$ with $G_0 = \bot=(0,0)$. Define
$f_1(x_1,x_2)=(1,x_2)$, and
$f_2(x_1,x_2)=(x_1,1).$
$L$ is finite, we start from the bottom element of $L$, and both functions are monotone and inflationary.
If the scheduler applies only $f_1$ forever (and never applies $f_2$), the execution is $
(0,0)\xrightarrow{f_1}(1,0)\xrightarrow{f_1}(1,0)\xrightarrow{f_1}\cdots
$
which stabilizes at $(1,0)$. However, $(1,0)$ is not a common fixed point since
$f_2(1,0)=(1,1)\neq(1,0)$. Thus, fairness is required.

The theorem assumes the execution starts from $G_0=\bot$. If the execution starts from an arbitrary state $G_0$, the same argument shows
that it converges to the \emph{least common fixed point above $G_0$}, which need
not be the least common fixed point of the entire lattice.

Monotonicity is necessary. Consider the lattice $L=\{0,1,2\}$ under the natural order.
Define $f_1(0)=2,\quad f_1(1)=1,\quad f_1(2)=2$, and
$f_2(0)=1,\quad f_2(1)=1,\quad f_2(2)=2$.
The function $f_1$ is not monotone but it is inflationary, while $f_2$ is monotone and inflationary.
The common fixed points of $f_1$ and $f_2$ are $\{1,2\}$, so the least common fixed
point is $1$. However, starting from $\bot=0$, applying $f_1$ followed by $f_2$
yields the state $2$, showing that without monotonicity the execution 
may converge to a common fixed point that is not the least one.

Finally, the inflationary assumption is required to guarantee the existence of a common fixed point.
Let $L=\{0,1\}$ under the natural order, and define
$f_1(x)=0$,and $f_2(x)=1$.
$f_2$ is monotone and inflationary while $f_1$ is monotone but not inflationary. Notice that neither $0$ nor $1$ are common fixed points in $L$.
Under a fair schedule that alternates between $f_1$ and $f_2$, the execution starting
from $0$ toggles between $0$ and $1$ indefinitely and does not terminate.

\end{proof}

\end{document}